\documentclass[smallcondensed]{svjour3}
\pdfoutput=1

\newcommand\Osoft{O^{\scriptscriptstyle \sim}\!}

\usepackage[utf8]{inputenc}
\usepackage[english]{babel}
\usepackage[T1]{fontenc}
\usepackage{IEEEtrantools}
\usepackage[colorlinks=true,citecolor=blue,linkcolor=blue,urlcolor=blue]{hyperref}
\usepackage{epsfig}
\usepackage{amsmath, amssymb, amsbsy}
\usepackage{xspace} 
\usepackage{color}
\usepackage{url}             
\usepackage{cite}            
\usepackage{booktabs} 
\usepackage{parskip} 
\usepackage{enumitem}
\usepackage{tikz}
\usepackage[linesnumbered,ruled,vlined,titlenumbered]{algorithm2e}
\usepackage{mathtools}
\usepackage{xparse}
\usepackage{url}
\usepackage[final,tracking=true,kerning=true,spacing=true,factor=1100,stretch=10,shrink=20]{microtype}

\title{Improved Power Decoding of Interleaved One-Point Hermitian Codes}

\author{Sven Puchinger \and Johan Rosenkilde \and Irene Bouw}
\authorrunning{Sven Puchinger, Johan Rosenkilde, Irene Bouw}
\institute{
Sven Puchinger \at Institute of Communications Engineering, Ulm University, Germany \\ \email{sven.puchinger@uni-ulm.de}
\and
Johan Rosenkilde \at Department of Applied Mathematics and Computer Science, Technical University of Denmark \\
\email{jsrn@jsrn.dk}
\and
Irene Bouw \at Institute of Pure Mathematics, Ulm University, Germany \\
\email{irene.bouw@uni-ulm.de}
}

\date{Received: date / Accepted: date}

\def\ve#1{{\mathchoice{\mbox{\boldmath$\displaystyle #1$}}%
              {\mbox{\boldmath$\textstyle #1$}}%
              {\mbox{\boldmath$\scriptstyle #1$}}%
              {\mbox{\boldmath$\scriptscriptstyle #1$}}}}

\newcommand{\printalgoIEEE}[1]
{{\centering
\scalebox{0.97}{
\begin{tabular}{p{\textwidth}}
\begin{algorithm}[H]
 #1
\end{algorithm}
\end{tabular}
}
}
}

\AtBeginDocument{}
\newcommand{\ZZ}{\mathbb{Z}}

\definecolor{sunset}{rgb}{1,0.5,.05}
\definecolor{marine}{rgb}{0,0,.7}
\definecolor{navy}{rgb}{0,0,.5}
\definecolor{forest}{rgb}{0,.6,0}
\definecolor{brown}{rgb}{0.59, 0.29, 0.0}

\renewcommand{\H}{\mathcal{H}}

\newcommand{\NN}{\mathbb{N}}

\renewcommand{\L}{\mathcal{L}}

\newcommand{\Fqm}{\mathbb{F}_{q^m}}

\renewcommand{\H}{\mathcal{H}}

\newcommand{\Fqtwo}{\mathbb{F}_{q^2}}

\newcommand{\Pinf}{{P_\infty}}
\newcommand{\Pset}{\mathcal{P}}
\newcommand{\Pstar}{\mathcal{P}^\ast}

\renewcommand{\r}{\ve r}
\renewcommand{\c}{\ve c}
\newcommand{\e}{\ve e}

\newcommand{\CHerm}{\mathcal{C}_\mathcal{H}}
\newcommand{\Eset}{\mathcal{E}}
\newcommand{\degH}{\deg_{\mathcal{H}}}

\newcommand{\Rback}{\mathcal{R}}

\newcommand{\vr}{{\ve \nu}}
\newcommand{\mr}{{\ve \mu}}
\newcommand{\XiM}{{\ve \Xi}}

\newcommand{\FX}{\Fqtwo[X]}

\newcommand{\AM}[1]{\Amatrix^{(#1)}}
\newcommand{\AMe}[1]{A^{(#1)}}

\newcommand{\Amatrix}{\ve A}

\newcommand{\ddesigned}{d^\ast}

\renewcommand{\u}{\ve u}

\newcommand{\Code}{\mathcal{C}}

\newcommand{\PfailPower}{\mathrm{\hat{P}}_\mathrm{fail, IPD}}
\newcommand{\PfailGS}{\mathrm{\hat{P}}_\mathrm{fail, GS}}

\renewcommand{\i}{{\ve i}}
\renewcommand{\j}{{\ve j}}
\newcommand{\0}{{\ve 0}}
\newcommand{\len}[1]{{|#1|}}

\definecolor{mygreen}{rgb}{0,0.5,0}
\definecolor{myred}{rgb}{0.7,0,0}
\newcommand{\IntDegree}{h}
\newcommand{\mH}{{m_\mathrm{H}}}

\newcommand{\f}{\ve f}
\newcommand{\R}{\ve R}
\newcommand{\OmegaVec}{\ve \Omega}
\renewcommand{\a}{\ve a}
\renewcommand{\b}{\ve b}

\renewcommand{\smaller}{\preceq}

\newcommand{\Jset}{\mathcal{J}}
\newcommand{\Iset}{\mathcal{I}}
\renewcommand{\th}{^\mathrm{th}}

\newcommand\modop{\ \textnormal{mod}\ }

\newcommand{\CountLt}[1]{\binom{\IntDegree+#1-1}{\IntDegree}}
\newcommand{\CountLeq}[1]{\binom{\IntDegree+#1}{\IntDegree}}

\newcommand{\WeightedLt}[1]{\IntDegree\binom{\IntDegree+#1-1}{\IntDegree+1}}
\newcommand{\CountLtt}[1]{\tbinom{\IntDegree+#1-1}{\IntDegree}}
\newcommand{\CountLeqt}[1]{\tbinom{\IntDegree+#1}{\IntDegree}}
\newcommand{\WeightedLeqt}[1]{\IntDegree\tbinom{\IntDegree+#1}{\IntDegree+1}}
\newcommand{\WeightedLtt}[1]{\IntDegree\tbinom{\IntDegree+#1-1}{\IntDegree+1}}

\newcommand{\NumberOfVariables}{\mathrm{NV}}
\newcommand{\NumberOfEquations}{\mathrm{NE}}

\newcommand{\taumax}{\tau_\mathrm{max}}
\newcommand{\decodingradius}{t_\mathrm{new}}

\newcommand{\round}[1]{\lfloor #1 \rfloor}

\newcommand{\Amat}{\ve{A}}

\newcommand{\IntDegreeRS}{\IntDegree_{\mathrm{RS}}}
\newcommand{\IntDegreeH}{\IntDegree_{\mathrm{H}}}

\newcommand{\Omegasi}{\Omega_{s,\i}}

\begin{document}

\maketitle

\begin{abstract}
We propose a new partial decoding algorithm for $\IntDegree$-interleaved one-point Hermitian codes that can decode---under certain assumptions---an error of relative weight up to $1-(\tfrac{k+g}{n})^{\frac{\IntDegree}{\IntDegree+1}}$, where $k$ is the dimension, $n$ the length, and $g$ the genus of the code.
Simulation results for various parameters indicate that the new decoder achieves this maximal decoding radius with high probability.
The algorithm is based on a recent generalization of Rosenkilde's improved power decoder to interleaved Reed--Solomon codes, does not require an expensive root-finding step, and improves upon the previous best decoding radius by Kampf at all rates.
In the special case $\IntDegree=1$, we obtain an adaption of the improved power decoding algorithm to one-point Hermitian codes, which for all simulated parameters achieves a similar observed failure probability as the Guruswami--Sudan decoder above the latter's guaranteed decoding radius.
\end{abstract}


\section{Introduction}

One-point Hermitian ($1$-H) codes are algebraic geometry codes that can be decoded beyond half the minimum Goppa distance.
Most of their decoders are conceptually similar to their Reed--Solomon (RS) code analogs, such as the \emph{Guruswami--Sudan} (GS) algorithm \cite{guruswami1998improved} and \emph{power decoding} (PD) \cite{schmidt2010syndrome,kampf2012decoding,nielsen2015sub}.
For both RS and $1$-H codes, PD is only able to correct as many errors as the Sudan algorithm, which is a special case of the GS algorithm.
Recently \cite{nielsen2016power}, PD for RS codes was improved to correct as many errors as the GS algorithm.

An $\IntDegree$-interleaved $1$-H code is a direct sum of $\IntDegree$ many $1$-H codes.
By assuming that errors occur at the same positions in the constituent codewords (burst errors), it is possible to decode far beyond half the minimum distance \cite{kampf2014bounds}, which is inspired by decoding methods for interleaved RS codes \cite{krachkovsky1997decoding,schmidt2010syndrome}.
In the RS case, there have been many improvements on the decoding radius in the last two decades \cite{krachkovsky1997decoding,parvaresh2004multivariate,schmidt2007enhancing,schmidt2010syndrome,cohn2013approximate,wachterzeh2014decoding,puchinger2017decoding}, which have not all been adapted to $1$-H codes. The currently best-known decoding radius for interleaved RS codes is achieved by both the interpolation-based technique in \cite{parvaresh2004multivariate,cohn2013approximate} and the method based on improved PD in \cite{puchinger2017decoding}, where the latter has a smaller complexity since it does not rely on an expensive root-finding step.

In this paper, we adapt the decoder in \cite{puchinger2017decoding}, which is based on improved PD, to $\IntDegree$-interleaved $1$-H codes using the description of PD for $1$-H as in \cite{nielsen2015sub}.
Similar to the RS case, we derive a larger system of non-linear key equations (cf.~Section~\ref{sec:key_equations}) and reduce the decoding problem to a linear problem whose solution---under certain assumptions---agrees with the solution of the system of key equations (cf.~Section~\ref{sec:pade_approximation}).

Using a linear-algebraic argument, we derive an upper bound on the maximum number of errors which can yield a unique solution of the linear problem (cf.~Section~\ref{sec:decoding_radius}).
This decoding radius improves upon the previous best, \cite{kampf2014bounds}, at all rates.
In Section~\ref{sec:numerical_results}, we present simulation results for various code and decoder parameters which indicate that the new algorithm achieves the maximal decoding radius with high probability.
The complexity of solving the linear problem is shown to be sub-quadratic in the code length in Section~\ref{sec:complexity}.
Finally, we compare the decoding radii of RS, interleaved RS and interleaved $1$-H codes for the same overall field size and length in Section~\ref{sec:comparison}.

In the special case $\IntDegree=1$, we obtain an $1$-H analogue of the improved PD for RS codes \cite{nielsen2016power}.
This improves the decoding radius of the $1$-H PD decoder of \cite{nielsen2015sub} at a similar cost, sub-quadratic in the code length, and similar to the best known cost of the $1$-H GS algorithm \cite{nielsen2015sub}.
Simulation results suggest that the decoder has a similar failure probability as the GS algorithm for the same parameters when the decoding radius is beyond the guaranteed radius of the GS algorithm (cf.~Section~\ref{sec:numerical_results}).

The decoder is described for codes of full length $n = q^3$; the approach works for any $n < q^3$, but to obtain the good complexities, certain restrictions to how the evaluation points are chosen should be kept.
For notational convenience, we restrict ourselves to homogeneous interleaved $1$-H codes, i.e., where the constituent codes have the same rate. The generalization to inhomogeneous codes is straightforward.

The results of this article were partly presented at the International Workshop on Coding and Cryptography, Saint-Petersburg, Russia, 2017, where we only considered the case $\IntDegree=1$ \cite{puchinger2017improved}.
While writing this extension, we discovered some slightly improved key equations, which are presented here.
In the previous paper we sought $\Lambda^s$ where $\Lambda$ is a usual notion of error-locator; now we instead define  and seek $\Lambda_s$, which is an ``error-locator of multiplicity $s$''.

\section{Preliminaries}

Let $q$ be a prime power. We follow the notation of \cite{nielsen2015sub}.
The \emph{Hermitian curve} $\H/\Fqtwo$ is the smooth projective plane curve defined by the affine equation $Y^q+Y=X^{q+1}$.
The curve $\H(\Fqtwo)$ has genus $g = \tfrac{1}{2}q(q-1)$ and $q^3+1$ many $\Fqtwo$-rational points $\Pset = \{P_1,\dots,P_{q^3},\Pinf\}$, where $\Pinf$ denotes the point at infinity.
We define $\Rback := \cup_{\mH\geq0} \L(\mH\Pinf) = \Fqtwo[X,Y]/(Y^q+Y-X^{q+1})$, which has an $\Fqtwo$-basis of the form $\{ X^iY^j : 0 \leq i, 0 \leq j < q\}$. The order function $\degH : \Rback \to \ZZ_{\geq 0} \cup \{-\infty\}, f \mapsto -v_\Pinf(f)$ is defined by the valuation $v_\Pinf$ at $\Pinf$. As a result, we have $\degH(X^iY^j) = i q + j(q+1)$.
We will think often operate with elements of $\Rback$ as bivariate polynomials in $X$ and $Y$, represented as $\Fqtwo$-linear  combinations of the aforementioned basis.
In this paper, when we say ``degree'' of an element in $\Rback$, we mean its $\degH$.
A non-zero element of $\Rback$ is called monic if its monomial of largest $\degH$ has coefficient~$1$.

Let $n=q^3$ and $\mH \in \NN$ with $2(g-1) < \mH <n$. The \emph{one-point Hermitian code} of length $n$ and parameter $\mH$ over $\Fqtwo$ is defined by
\begin{align*}
\CHerm(n,\mH) = \left\{ \left( f(P_1) , \dots , f(P_n) \right) : f \in \L(\mH\Pinf) \right\}.
\end{align*}
The dimension of $\CHerm$ is given by $k=\mH-g+1$ and the minimum distance $d$ is lower-bounded by the \emph{designed minimum distance} $\ddesigned := n-\mH$.

The \emph{(homogeneous) $\IntDegree$-interleaved one-point Hermitian code} of length $n$ and parameter $\mH$ over $\Fqtwo$ is the direct sum of $\IntDegree$ one-point Hermitian codes $\CHerm(n,\mH)$, i.e.,
\begin{align*}
\CHerm(n,\mH;\IntDegree) = \left\{ \begin{bmatrix}
\c_1 \\
\vdots \\
\c_\IntDegree
\end{bmatrix} \in \Fqtwo^{\IntDegree \times n} \, : \,  \c_i \in \CHerm(n,\mH) \right\}.
\end{align*}
As a metric for errors, we consider \emph{burst errors}: If $\r=\c+\e \in \Fqtwo^{\IntDegree \times n}$ is received for a codeword $\c \in \CHerm(n,\mH;\IntDegree)$, then the \emph{error positions} $\Eset \subseteq \{1,\dots,n\}$ are given by the non-zero columns of $\e$, i.e.,
\begin{align*}
\Eset := \textstyle\bigcup_{j=1}^{\IntDegree} \left\{ i \, : \, e_{j,i} \neq 0 \right\}.
\end{align*}

For a vector $\i = [i_1,\ldots, i_m] \in \ZZ_{\geq 0}^\IntDegree$, we define its \emph{size} as $|\i| := \sum_\mu i_\mu$.
We denote by $\smaller$ the product partial order on $\ZZ_{\geq 0}^\IntDegree$, i.e.~$\i \smaller \j$ if $i_\mu \leq j_\mu$ for all $\mu$.
The number of vectors $\ve \in \ZZ_{\geq 0}^m$ of size $\len{\i} = \mu$ is given by $\tbinom{\IntDegree+\mu-1}{\mu}$.
We use the following relations, which follow from properties of the binomial coefficient.
\begin{lemma}\label{lem:binomial_sums}
Let $m,t \in \ZZ_{>0}$. Then,
\begin{align*}
\textstyle\sum_{\mu=0}^{t} \tbinom{m+\mu-1}{\mu} = \tbinom{m+t}{m}, \text{ and }
\textstyle\sum_{\mu=0}^{t-1} \mu \tbinom{m+\mu-1}{\mu} = t \tbinom{m+t-1}{m+1}.
\end{align*}
\end{lemma}
Note that the Lemma~\ref{lem:binomial_sums} means e.g.
\begin{equation*}
  \sum_{\i \in \ZZ^h_{\geq 0}, \len{\i} < t} \len{\i} = t \tbinom{h+t-1}{h+1} \ .
\end{equation*}
We also introduce the following notational short-hands:
\begin{definition}
For $\a \in \Rback^\IntDegree$, and $\i,\j \in \ZZ_{\geq 0}^\IntDegree$, we define
\begin{align*}
\a^\i := \textstyle\prod_{\mu=1}^{\IntDegree} a_\mu^{i_\mu}, \quad
\tbinom{\j}{\i} := \textstyle\prod_{\mu=1}^{m} \tbinom{j_\mu}{i_\mu}.
\end{align*}
\end{definition}
By extending the binomial theorem to this notation, we obtain the following lemma.
\begin{lemma}\label{lem:generalized_binomial}
Let $\a,\b \in \Rback^\IntDegree$, and $\j \in \ZZ_{\geq 0}^m$. Then,
\begin{align*}
(\a+\b)^\j = \textstyle\sum_{\i \smaller \j} \tbinom{\j}{\i} \a^\i \b^{\j-\i}.
\end{align*}
\end{lemma}

For computational complexities, we use the soft-$O$ notation $O^\sim$, which omits log factors.

\section{System of Key Equations}
\label{sec:key_equations}

In this section, we derive the system of key equations that we need for decoding, using the same trick as \cite{puchinger2017decoding} for interleaved Reed--Solomon codes.
We use the description of power decoding for one-point Hermitian codes as in \cite{nielsen2015sub}.
Suppose that the received word is $\r = \c + \e \in \Fqtwo^{\IntDegree \times n}$, consisting of an error $\e$ with corresponding (burst) error positions $\Eset$ and a codeword $\c \in \CHerm(n,\mH;\IntDegree)$, which is obtained from the \emph{message polynomials} $\f = [f_1,\dots,f_\IntDegree] \in \L(m\Pinf)^\IntDegree$.

In the following sections we show how to retrieve the message polynomials $\f$ from the received word $\r$ if the \emph{number of errors} $|\Eset|$ does not exceed a certain decoding radius, which depends on the parameters of the decoding algorithm.
Similar to \cite{nielsen2015sub}, we define the following polynomials.

\begin{definition}
Let $s \in \NN$. The \emph{error locator polynomial $\Lambda_s$ of multiplicity $s$} is the element in $\L\left(-\sum_{i \in \Eset} s P_i + \infty \Pinf\right)$ of minimal degree that is monic.
\end{definition}

\begin{theorem}
The error locator polynomial of multiplicity $s$ is unique and has degree
\begin{align*}
s |\Eset| \leq \degH \Lambda_s \leq s |\Eset|+g.
\end{align*}
\end{theorem}

\begin{proof}
The proof is similar to \cite[Lemma~23]{nielsen2015sub}.
Uniqueness is clear since if there were two such polynomials, their difference would also be in $\L\left(-\sum_{i \in \Eset} s P_i + \infty \Pinf\right)$, but of smaller $\degH$.
Being in $\L\left(-\sum_{i \in \Eset} s P_i + \infty \Pinf\right)$ specifies $s |\Eset|$ homogeneous linear equations in the coefficients of $\Lambda_s$, since for any $i \in \Eset$, we can expand $\Lambda_s$ into a power series $\sum_{j \geq s} \gamma_{i,j} \phi_i^j$ for a local parameter $\phi_i$ of $P_i$ (e.g., take $\phi_i = X-\alpha_i$ if $P_i = (\alpha_i,\beta_i)$).
By requiring $\degH \Lambda_s \leq s |\Eset|+g$, we have more variables than equations, so there is a non-zero $\Lambda_s$ of the sought form with degree at most $s |\Eset|+g$.
The lower bound works exactly as in \cite[Lemma~23]{nielsen2015sub}.
\end{proof}

\begin{lemma}
\label{lem:IP}
For each $i=1,\dots,\IntDegree$, there is a polynomial $R_i \in \Rback$ with $\degH(R_i)<n+2g$ that satisfies $R(P_j) = r_{i,j}$ for all $P_j \in \Pstar$.
Each $R_i$ can be computed in $O^\sim(n)$ operations over $\Fqtwo$.
\end{lemma}

\begin{proof}
Apply \cite[Lemma~6]{nielsen2015sub} to each row of the received word.
\end{proof}

In the following, let $\R = [R_1,\dots,R_\IntDegree] \in \Rback^\IntDegree$ be as in Lemma~\ref{lem:IP} and $G \in \Rback$ be defined as
\begin{align*}
G = \textstyle\prod_{\alpha \in \Fqtwo} (X-\alpha) = X^{q^2}-X.
\end{align*}

\begin{lemma}\label{lem:omega_s_i}
For each $\i \in \ZZ_{\geq 0}$ with $\len{\i} \leq s$, there is a unique $\Omegasi \in \Rback$ of degree $\degH \Omegasi \leq \degH \Lambda_s + \len{\i} (2g-1)$ such that
\begin{align*}
\Lambda_s (\f-\R)^\i = G^{\len{\i}} \Omegasi.
\end{align*}
\end{lemma}

\begin{proof}
Consider $v_{P_j}(\Lambda_s (\f-\R)^\i)$ for $j = 1,\ldots,n$: if $j \in \Eset$ then $v_{P_j}(\Lambda_s (\f-\R)^\i) = v_{P_j}(\Lambda_s) \geq s \geq \len{\i}$.
If $j \notin \Eset$ then $v_{P_j}(\Lambda_s (\f-\R)^\i) \geq v_{P_j}((\f-\R)^\i) \geq \len{\i}$.
We conclude
\begin{align*}
\Lambda_s (\f-\R)^\i \in \L\Big(-\len{\i} \sum_{j=1}^{n} P_j + \infty \Pinf\Big).
\end{align*}
Since the divisor in that $\L$-space is exactly $\mathrm{div}(G^{\len{\i}}) + \infty \Pinf$, then $\Lambda_s (\f-\R)^\i$ must be divisible by $G^{\len{\i}}$ (see e.g., \cite[Lemma~3]{nielsen2015sub}) with quotient in $\Rback$.
The degree is given by taking $\degH$ on both sides and using $\degH(R_i) < n+2g-1$.
\end{proof}

The following theorem states the system of key equations that we will use for decoding in the next sections.
Note that the formulation is similar to its interleaved Reed--Solomon analog \cite{puchinger2017decoding}, with the difference that all involved polynomials are elements of the ring $\Rback$.

\begin{theorem}[System of Key Equations]\label{thm:key_equations}
Let $\ell,s \in \ZZ_{>0}$ be such that $s \leq \ell$ and $\Lambda_s$, $\f$, $\R$, $G$, and $\Omega_{s,\i}$ as above. Then, for all $\j \in \ZZ_{\geq 0}^\IntDegree$ of size $1 \leq \len{\j} \leq \ell$, we have
\begin{align}
\Lambda_s \f^\j
&= \sum\limits_{\i \smaller \j} \Omegasi \left[ \binom{\j}{\i} \R^{\j-\i} G^{\len{\i}}\right], && 1 \leq \len{\j}<s \label{eq:key_eq_equality} \\
\Lambda_s \f^\j &\equiv \sum\limits_{\substack{\i \smaller \j \\ \len{\i} < s}}
\Omega_{s,\i}
\left[ \binom{\j}{\i} \R^{\j-\i} G^{\len{\i}}\right]
  \mod G^s, &&s \leq \len{\j} \leq \ell, \label{eq:key_eq_congruence}
\end{align}
as congruences over $\Rback$.
\end{theorem}

\begin{proof}
Using Lemma~\ref{lem:generalized_binomial}, we obtain
\begin{align}
\Lambda_s \f^\j = \Lambda_s \left(\R + \left(\f-\R\right)\right)^\j = \textstyle\sum_{\i \smaller \j} \tbinom{\j}{\i} \Lambda_s \left(\f-\R\right)^\i \R^{\j-\i}. \label{eq:key_eq_proof_eq1}
\end{align}
In all summands with $\len{\i} <s$, we can rewrite, using Lemma~\ref{lem:omega_s_i},
\begin{align}
\Lambda_s \left(\f-\R\right)^\i = G^\len{\i} \Omegasi. \label{eq:key_eq_proof_eq2}
\end{align}
If $\len{\i} \geq s$, we can write $\i = \i'+\i''$, for some $\i',\i'' \in \ZZ_{\geq 0}$ with $\len{\i'} = s$, and
\begin{align*}
\Lambda_s \left(\f-\R\right)^\i = \Lambda_s(\f-\R)^{\i'} (\f-\R)^{\i''} = G^s \Omega_{s,\i'} (\f-\R)^{\i''},
\end{align*}
so all those terms are divisible by $G^s$. For $\len{\j} < s$, all summands of \eqref{eq:key_eq_proof_eq1} have $\len{\i} \leq \len{\j} < s$ and are of the form \eqref{eq:key_eq_proof_eq2}. We therefore obtain \eqref{eq:key_eq_equality}. For $\len{\j} \geq s$, all summands of \eqref{eq:key_eq_proof_eq1} with $\len{\i}\geq s$ are divisible by $G^s$, so we get \eqref{eq:key_eq_congruence}.
\end{proof}

\section{Solving the System of Key Equations}
\label{sec:pade_approximation}

The key equations in Theorem~\ref{thm:key_equations} are non-linear relations between the unknown polynomials $\Lambda_s$, $\f$, and $\OmegaVec$.
We therefore relax them into---at the first glance much weaker---linear problem and hope that their solutions agree.
The resulting problem is a heavy generalisation of multi-sequence linear shift register synthesis \cite{feng_generalized_1989,nielsen_generalised_2013}, which is very related to simultaneous Hermite Pad\'e approximations \cite{beckermann_uniform_1994}.

\begin{problem}\label{prob:IH_decoding_SRP}
  Consider a code $\Code = \CHerm(n,\mH;\IntDegree)$ and a decoding instance with received word $\r = \c + \e \in \Fqtwo^{h \times n}$, where $\c \in \Code$ is unknown and is obtained from the unknown message polynomials $\f \in \L(m\Pinf)^\IntDegree$.
  Let $\R$ and $G$ be as in Section~\ref{sec:key_equations}.
  Given positive integers $s \leq \ell$, let
\begin{align*}
A_{\i,\j} = \tbinom{\j}{\i} \R^{\j-\i} G^\len{\i} \in \Rback
\end{align*}
for all $\i \in \Iset := \{\i \in \NN_0^\IntDegree : 0 \leq \len{\i} < s\}$ and $\j \in \Jset := \{\j \in \NN_0^\IntDegree : 1 \leq \len{\j} \leq \ell\}$. Find $\lambda_\i, \psi_\j \in \Rback$ for $\i \in \Iset$ and $\j \in \Jset$ with monic $\lambda_\0$, such that
\begin{align}
\psi_\j &= \sum\limits_{\i \in \Iset} \lambda_\i A_{\i,\j} &&\j \in \Jset \textrm{ and } \len{\j} < s, \label{eq:IH_decoding_SRP_equation} \\
\psi_\j &\equiv \sum\limits_{\i \in \Iset} \lambda_\i A_{\i,\j} \mod G^s &&\j \in \Jset \textrm{ and } \len{\j} \geq s, \label{eq:IH_decoding_SRP_congruence} \\
\degH \lambda_\0 &\geq \degH \lambda_\i-\len{\i}(2g-1) && \i \in \Iset, \label{eq:IH_decoding_SRP_deg_lambdai} \\
\degH \lambda_\0 &\geq \degH \psi_\j - \len{\j}\mH &&  \j \in \Jset . \label{eq:IH_decoding_SRP_deg_psii}
\end{align}
\end{problem}

\begin{definition}
Consider an instance of Problem~\ref{prob:IH_decoding_SRP}.
We say that a solution $(\lambda_\i)_{i \in \Iset},(\psi_\j)_{\j \in \Jset}$, has \emph{degree} $\tau \in \ZZ_{\geq 0}$ if $\degH \lambda_\0 = \tau$.
Furthermore, we call a solution \emph{minimal} if its degree is minimal among all solutions.
\end{definition}

Problem~\ref{prob:IH_decoding_SRP} is connected to the key equations through the following statement.

\begin{theorem}\label{thm:IH_decoding_SRP_Lambda_Psi_solution}
Consider an instance of Problem~\ref{prob:IH_decoding_SRP}. Then,
\begin{align*}
\lambda_\i &= \Lambda_\i := \Omegasi, &&\i \in \Iset, \\
\psi_\j &= \Psi_\j := \Lambda_s \f^\j, &&\j \in \Jset,
\end{align*}
is a solution to the problem of degree $\tau = \degH \Lambda_s$, where $s \cdot |\Eset| \leq \tau \leq s \cdot |\Eset|+g$.
\end{theorem}
\begin{proof}
  Note $\Omega_{s,\0} = \Lambda_s$.
  The equalities and congruences are now clear from the key equations.
  As for the degree restrictions, we have
\begin{align*}
\degH \Lambda_\i &\leq \degH \Lambda_\0 + \len{\i}(2g-1), \\
\degH \Psi_\j &= \degH(\Lambda_s) + \degH(\f^\j) \leq \degH \Lambda_\0 + \len{\j}\mH,
\end{align*}
which proves the claim.
\end{proof}

\begin{remark}
Most received words will satisfy $\degH R_i = n+2g-1$ for all $i=1,\dots,\IntDegree$.
In such a case, the solution of Problem~\ref{prob:IH_decoding_SRP} given in Theorem~\ref{thm:IH_decoding_SRP_Lambda_Psi_solution} fulfills all degree restrictions of the problem with equality.
These relative upper bounds on the degrees of $\lambda_\i$ and $\psi_\j$ are therefore the minimal choice among all such bounds for which Theorem~\ref{thm:IH_decoding_SRP_Lambda_Psi_solution} holds.
\end{remark}

Theorem~\ref{thm:IH_decoding_SRP_Lambda_Psi_solution} motivates a decoding strategy, which is outlined in Algorithm~\ref{alg:IH_improved_power}:
To every codeword $\c' \in \CHerm(n,\mH;\IntDegree)$ corresponds a solution to Problem~\ref{prob:IH_decoding_SRP} whose degree is roughly $s \cdot |\Eset'|$, where $|\Eset'|$ is the number of errors (i.e., non-zero columns) of $\r-\c'$.
Among those solutions, we want to find the one of smallest degree, i.e., the one for the closest codeword.
There will also be other solutions to Problem~\ref{prob:IH_decoding_SRP}, which do not correspond to codewords, but the idea is that in most cases, and when the number of errors is not too large, the minimal solution \emph{will} correspond to the closest codeword.

\printalgoIEEE{
\DontPrintSemicolon
\KwIn{Received word $\r \in \Fqm^{\IntDegree \times n}$ and positive integers $s \leq \ell$}
\KwOut{$\f \in \L(\mH \Pinf)^\IntDegree$ such that $\c_i = [f_i(P_1),\dots,f_i(P_n)]$ for all $i=1,\dots,\IntDegree$ is the codeword with a corresponding minimal $\degH \Lambda_s$; or ``decoding failure''.}
Compute $\R$ and $G$ as in Section~\ref{sec:key_equations} \;
$A_{\i,\j} \gets \tbinom{\j}{\i} \R^{\j-\i} G^{\len{\i}}$ for all $\i \smaller \j$ \;
$\lambda_\i,\psi_\j \gets$ Minimal solution to Problem~\ref{prob:IH_decoding_SRP} with input $s$, $\ell$, $A_{\i,\j}$, and $G$ \label{line:minimal_solution}\;
\If{$\lambda_\0$ divides all $\psi_{\u_i}$ over $\Rback$ for $i=1,\dots,\IntDegree$, where $\u_i$ is the $i\th$ unit vector}{
$\f \gets [\psi_{\u_1}/\lambda_\0,\dots,\psi_{\u_\IntDegree}/\lambda_\0]$ \;
$\Eset \gets $ Error set corresponding to $\e_i = \r_i-[f_i(\alpha_1),\dots,f_i(\alpha_n)]$ for $i=1,\dots,\IntDegree$ \;
\If{$\lambda_\0 \in \L\left(-\sum_{i \in \Eset} s P_i + \infty \Pinf\right)$ {\rm \textbf{and}} $s \cdot |\Eset| \leq \degH \lambda_\0 \leq s (|\Eset|+g)$}{
\Return{$\f$}
}
}
\Return{``decoding failure''}
\caption{Improved Power Decoder for $\IntDegree$-Interleaved 1-Point Hermitian Codes}
\label{alg:IH_improved_power}
}

In the cases for which this does not happen, the decoder will fail; we will return to this in Section~\ref{sec:decoding_radius}.
If the algorithm finds a solution that corresponds to a codeword, then we have $\lambda_\0 = \Lambda_s$ and $\psi_{\u_i} = \Lambda_s f_i$ for $i=1,\dots,\IntDegree$, where $\u_i = [0,\dots,1,\dots,0]$ is the $i\th$ unit vector.
Hence, we obtain the $i\th$ message polynomial $f_i$ by division of $\psi_{\u_i}$ by $\lambda_\0$.

Note that Algorithm~\ref{alg:IH_improved_power} does not exactly promise to find the closest codeword: it finds the codewords whose corresponding $\Lambda_s$ has minimal $\degH$.
When the number of errors is very small, we will often or always have $\degH \Lambda_s < s|\Eset| + g$; but in this case all other codewords are much farther away from $\r$.
On the other hand, when the number of errors is large, most error vectors will satisfy $\degH \Lambda_s = s|\Eset| + g$.
In both these cases Algorithm~\ref{alg:IH_improved_power} will find the closest codewords.
It seems reasonable to expect, however, that there exist some rare received words for which a farther codeword will have an associated $\Lambda_s$ of lower $\degH$ than the closest codeword.

We will see in Section~\ref{sec:complexity} that we can find a minimal solution of Problem~\ref{prob:IH_decoding_SRP} efficiently.

\section{Decoding Radius and Failure Behavior}
\label{sec:decoding_radius}

In this section, we derive an upper bound on the maximal degree of the error locator polynomial $\Lambda_s$ for which there can be a unique minimal solution of Problem~\ref{prob:IH_decoding_SRP}.
Since the degree of $\Lambda_s$ is related to the number of errors, this implies an estimate of the maximal decoding radius of our decoder.
We also briefly discuss in which cases the decoder fails below this bound.

\begin{lemma}\label{lem:delta_tau_homogeneous}
Let $\tau,\ell,s \in \NN$ such that $s \leq \ell$ and $\tau+\ell \mH < sn$.
All polynomials $\lambda_\i,\psi_\j \in \Rback$ for $\i \in \Iset$ and $\j \in \Jset$ that fulfill \eqref{eq:IH_decoding_SRP_equation}, \eqref{eq:IH_decoding_SRP_congruence}, and the absolute degree restrictions
\begin{align}
\degH \lambda_\i - \len{\i} (2g-1) &\leq \tau, \label{eq:IH_SRP_absolute_degree_1}\\
\degH \psi_\j - \len{\j} \mH &\leq \tau, \label{eq:IH_SRP_absolute_degree_2}
\end{align}
can be computed by a homogeneous linear system of equations over $\Fqtwo$ with at least
\begin{align*}
\delta(\tau) = (\tau+1) \CountLeqt{\ell} - n \left[ \WeightedLtt{s} +  s\CountLeqt{\ell} - s \CountLtt{s} \right] + \mH\WeightedLeqt{\ell} - g \CountLeqt{\ell}
\end{align*}
more variables than equations, whenever $\delta(\tau) \geq 0$.
If $\tau \geq 2g-1$, there are received words for which the difference is exactly $\delta(\tau)$.
\end{lemma}

\begin{proof}
We have $\degH A_{i,j} \leq (n+2g-1) \len{\j} - (2g-1)\len{\i}$, so we get
\begin{align*}
\degH\Big(\sum_{\i \in \Iset} \lambda_\i A_{\i,\j}\Big) &\leq \tau + \len{\j} (n+2g-1) \quad \forall \, \j \in \Jset.
\end{align*}
Thus, for most $\j$ the polynomial $\psi_\j$ has lower degree than the terms in $\sum_{\i \in \Iset} \lambda_\i A_{\i,\j}$ in the case $\len{\j} < s$ and less than the degree of the modulus $G^s$ in the other case.
Consider functions in $\Rback$ over the basis $\{ X^i Y^j\}$ over $\Fqtwo$.
Since the $\Fqtwo$-coefficients of $\sum_{\i \in \Iset} \lambda_\i A_{\i,\j}$ and $(\sum_{\i \in \Iset} \lambda_\i A_{\i,\j} \modop G^s)$ are known linear combinations of the unknown coefficients of the $\lambda_\i$, the restrictions of the lemma on the degrees of $\psi_\j \in \Rback$ can be described by an $\Fqtwo$ homogeneous linear system of equations that specify that the top coefficients of $\psi_\j$ be zero ($\tau + \len{\j}\mH+1$ and higher).

For non-negative integers $a$ and $b$, there are between $b-a-g$ and $b-a$ many monomials $x^iy^j \in \Rback$ with $j<q$ of degree at least $a$ and less than $b$. The lower bound is due to the Riemann--Roch theorem and the upper bound follows from the injectivity of $\degH$ on the set of monomials.

Due to the degrees of the involved polynomials, the number of $\Fqtwo$-linear restrictions for each $\len{\j} < s$ becomes
\begin{align*}
N_\j = (\tau + \len{\j}(n+2g-1)) - (\tau + \len{\j} \mH) = \len{\j} (n+2g-1-\mH).
\end{align*}
For $\len{\j} \geq s$, the analysis is a bit more involved: Since $G$ is a polynomial only in $X$ with $\deg_X(G^s) = s q^2$, the congruence modulo $G^s$ reduces the $X$-degree of all monomials below $s q^2$, i.e., the polynomial $(\sum_{\i \in \Iset} \lambda_\i A_{\i,\j} \modop G^s)$ can be written as
\begin{align*}
\left(\sum_{\i \in \Iset} \lambda_\i A_{\i,\j} \modop G^s\right) = \sum_{j=0}^{q-1} \sum_{i=0}^{s q^2-1} a_{ij} X^i Y^j,
\end{align*}
where $a_{ij} \in \Fqtwo$ are linear expressions of the coefficients of the $\lambda_\i$. By the degree restriction of $\psi_\j$, we must have that the coefficients $a_{ij}$ with $\degH(X^iY^j) > \tau+\len{\j} \mH$ are zero. Thus, we get at most
\begin{align*}
N_\j &= \left(\sum_{j=0}^{q-1} s q^2 \right) - \underset{= \, \tau+\len{\j} \mH -g+1}{\underbrace{|\{(i,j) \, : \, q i + (q+1)j \leq \tau+\len{\j} \mH\}|}} \\
&= s n - \tau -\len{\j} \mH + g - 1
\end{align*}
linear equations.
Note that the condition $\tau+\ell \mH < sn$ guarantees that there is no monomial $X^iY^j$ of $\degH(X^iY^j) > \tau +\len{\j} \mH$ with $\deg_X\geq sq^2$.
In total, and using Lemma~\ref{lem:binomial_sums} repeatedly:
\begin{align*}
\NumberOfEquations &= \sum\limits_{\j \in \Jset} N_\j = \sum\limits_{1 \leq \len{\j} < s} \len{\j}(n+2g-1-\mH) + \sum\limits_{s \leq \len{\j} \leq \ell} \big(sn - \tau -\len{\j} \mH + g - 1\big) \\
&= (n+2g-1)\WeightedLtt{s} + \left(s n-\tau-1 + g\right)\left(\CountLeqt{\ell} - \CountLtt{s}\right) - \mH\WeightedLeqt{\ell}.
\end{align*}
The number of variables, i.e., the number of $\Fqtwo$-coefficients of the $\lambda_\i$ is at least
\begin{align*}
\NumberOfVariables = \Big(\sum\limits_{\i \in \Iset}(\tau+\len{\i}(2g-1) + 1-g) \Big) = (\tau+1-g) \CountLtt{s} +(2g-1) \WeightedLtt{s}.
\end{align*}
The claim follows by subtracting $\NumberOfVariables-\NumberOfEquations$.

In the case $\tau \geq 2g-1$, all Weierstra{\ss} gaps are below the degree bounds of the $\lambda_\i$ and $\psi_\j$. Hence, the number of variables and equations is equal to the derived $\NumberOfEquations$ and $\NumberOfVariables$, respectively, as long as the maximal possible degree of $\degH\Big(\sum_{\i \in \Iset} \lambda_\i A_{\i,\j}\Big)$, i.e., for some choice of the $\lambda_\i$, is equal to $\tau + \len{\j} (n+2g-1)$.
There are received words for which $\degH R_i = n+2g-1$ for all $i$. In these cases, we can have $\degH\Big(\sum_{\i \in \Iset} \lambda_\i A_{\i,\j}\Big) = \tau + \len{\j} (n+2g-1)$ for some values of $\lambda_\i$, so (if $\tau \geq 2g-1$), the number of variables minus the number of equations is exactly $\delta(\tau)$.
\end{proof}

\begin{lemma}\label{lem:delta_tau}
If Problem~\ref{prob:IH_decoding_SRP} has a solution of degree $\tau$, it has at least $(q^2)^{\delta(\tau)-1}$ many such solutions.
\end{lemma}

\begin{proof}
Solutions of degree $\tau$ of Problem~\ref{prob:IH_decoding_SRP} are exactly the solutions of the homogeneous linear system in Lemma~\ref{lem:delta_tau_homogeneous} with $\degH \lambda_\0 = \tau$ and monic $\lambda_\0$. Thus, we set the $\tau\th$ coefficient of $\lambda_\0$ to $1$ and obtain an inhomogeneous linear system of equations with at least $\delta(\tau)-1$ more variables than equations.
If Problem~\ref{prob:IH_decoding_SRP} has a solution of degree $\tau$, then this system has at least $(q^2)^{\ker(\Amat)}$ solutions, where $\Amat$ is the system's matrix.
The claim follows by $\dim(\ker(\Amat)) \geq \delta(\tau)-1$.
\end{proof}

Lemma~\ref{lem:delta_tau} implies the following statement.

\begin{theorem}\label{thm:taumax}
Let $\tau = \degH \Lambda_s$ and $s,\ell \in \NN$ such that $s \leq \ell$ and $\tau+\ell \mH < sn$ and
\begin{align}
\tau > \taumax := sn \left( 1 - \tfrac{s \CountLt s - \WeightedLt s}{s\CountLeq \ell} \right) - \tfrac \IntDegree {\IntDegree+1} \ell \mH + \left(\tfrac{1}{\CountLeq \ell}-1\right) + g. \label{eq:IH_SRP_at_least_one_solution}
\end{align}
Then, Problem~\ref{prob:IH_decoding_SRP} has at least two solutions of degree $\tau$.
\end{theorem}

\begin{proof}
Condition \eqref{eq:IH_SRP_at_least_one_solution} is fulfilled if and only if $\delta(\tau)>1$.
Due to $\tau = \degH \Lambda_s$, Problem~\ref{prob:IH_decoding_SRP} has a solution of degree $\tau$ and the claim follows by Lemma~\ref{lem:delta_tau}.
\end{proof}

Theorem~\ref{thm:taumax} can be interpreted as follows: If $\degH \Lambda_s > \taumax$, then either $\Lambda_s$ does not correspond to a minimal solution of Problem~\ref{prob:IH_decoding_SRP}, or it is a minimal solution but there are many other minimal solutions as well.
There is no reason to think that our solver for Problem~\ref{prob:IH_decoding_SRP} will find $\Lambda_s$ among all those solutions, so decoding will likely fail.

Since we have $|\Eset| \leq \degH \Lambda_s \leq s|\Eset|+g$, and often $\degH \Lambda_s = s|\Eset|+g$, we usually have no unique solution whenever $s |\Eset| + g > \taumax$, i.e.,
\begin{align*}
|\Eset| > \tfrac{\taumax-g}{s}
\end{align*}
and for sure if $|\Eset| > \tfrac{\taumax}{s}$.
We therefore call
\begin{align}
\decodingradius &= \tfrac{\taumax-g}{s} = n \left( 1 - \tfrac{s \CountLt s - \WeightedLt s}{s\CountLeq \ell} \right) - \tfrac \IntDegree {\IntDegree+1} \tfrac{\ell}{s} \mH + \tfrac{1}{s}\left(\tfrac{1}{\CountLeq \ell}-1\right) \label{eq:decoding_radius}
\end{align}
the \emph{decoding radius} of Algorithm~\ref{alg:IH_improved_power}.

\begin{remark}\label{rem:weierstrass_gaps_special_case}
For $\tau \geq 2g-1$, by Lemma~\ref{lem:delta_tau_homogeneous}, there are received words (in fact most of them) such that the difference of numbers of variables and equations of the inhomogeneous system for computing all degree-$\tau$ solutions of Problem~\ref{prob:IH_decoding_SRP} is exactly $\delta(\tau)-1$.
Thus, if there are sufficiently many linearly independent equations\footnote{This linear-algebraic condition resembles, but seems weaker than, the ``(non-linear) algebraic independence assumption'' in \cite{cohn2013approximate} for decoding interleaved RS codes.}, there is no other solution of the problem, besides the error locator, of degree $\tau$ whenever $\delta(\tau)<0$.

For $\tau < 2g-1$, the degree bounds of $\lambda_\0$ and $\psi_\0$ are smaller than $2g-1$, but those of all other $\lambda_\i$ and $\psi_\j$ are bigger (note that $\mH \geq 2g-1$).
Thus, there can be up to $g$ fewer equations and up to $g$ more variables than predicted by $\delta(\tau)$ for any received word. The value of $\taumax$ as in Theorem~\ref{thm:taumax} can in this case therefore be smaller by a value up to
\begin{align*}
\taumax' = \taumax - 2g \tfrac{1}{\binom{\IntDegree + \ell}{\IntDegree}},
\end{align*}
which reduces the decoding radius by at most $2g/[s \binom{\IntDegree + \ell}{\IntDegree}]$.

In the case $\tau+\ell \mH \geq sn$, the number of equations is also smaller than predicted by $\taumax$. However, we will see in Section~\ref{ssec:asymptotic_anaylsis} that the best choice of $s$ for a given $\ell$ yields $\tau+\ell \mH < sn$ for $\tau \leq \taumax$.
\end{remark}

Since we cannot guarantee that the linear equations of the system in Lemma~\ref{lem:delta_tau_homogeneous} are linearly independent for $\tau \leq \degH \Lambda_s$, Algorithm~\ref{alg:IH_improved_power} can fail to return the sent codeword $\c$ for some errors of weight less than the maximal decoding radius.
In these cases, we have one of the following.
\begin{itemize}
\item There is a solution of Problem~\ref{prob:IH_decoding_SRP} of degree $< \degH \Lambda_s$.
\item There is more than one solution of Problem~\ref{prob:IH_decoding_SRP} of degree $= \degH \Lambda_s$ and the decoder picks the wrong one.
\end{itemize}
However, the simulation results for various code and decoder parameters, presented in the following section, indicate that the new decoder is able to decode most error patterns up to the derived decoding radius $\decodingradius$.
Sometimes, decoding succeeds even beyond $\decodingradius$.
In these cases, we usually have $\degH(\Lambda_s)< s|\Eset|+g$.

In all previous power decoding algorithms for Reed--Solomon \cite{schmidt2010syndrome,nielsen2016power}, one-point Hermitian \cite{kampf2014bounds,nielsen2015sub}, and interleaved Reed--Solomon codes \cite{puchinger2017decoding}, simulation results indicate that the failure probability for a number of errors below the maximal decoding radius is small and decreases exponentially in the difference of maximal decoding radius and number of errors.

As for these other variants, except for a few parameters of theirs (e.g., $\ell \leq 3$ and $s \leq 2$ for a single Reed--Solomon code in \cite{nielsen2016power}), it remains an open problem to prove an analytic upper bound on the failure probability of Algorithm~\ref{alg:IH_improved_power}.

\subsection{Asymptotic Analysis and Parameter Choice}
\label{ssec:asymptotic_anaylsis}

We study the asymptotic behavior of the decoding radius $\taumax$ and give explicit parameters to achieve the given limit.
The analysis is based on the following lemma.

\begin{lemma}[\!\!{\cite[Lemma~14]{puchinger2017decoding}}]\label{lem:binomial_ratio_convergence}
Let $\gamma \in (0,1)$ and $\IntDegree \in \NN$ be fixed.
Then, we have
\begin{align*}
\tfrac{\binom{\IntDegree+\round{\gamma i}}{\IntDegree}}{\binom{\IntDegree+i}{\IntDegree}} = \gamma^\IntDegree + O(\tfrac{1}{i}) \quad \text{ for } (i \to \infty).
\end{align*}
\end{lemma}

\begin{theorem}\label{thm:IH_asymptotic_relative_decoding_radius}
Let $(\ell_i,s_i) = (i, \round{\gamma i}+1)$ for $i \in \NN$, where $\gamma = \sqrt[\IntDegree+1]{\frac{\mH}{n}}$. Then,
\begin{align*}
\decodingradius(\ell_i,s_i) = n \Big( 1-\left(\tfrac{\mH}{n}\right)^{\frac{\IntDegree}{\IntDegree+1}} - O(\tfrac{1}{i}) \Big) \quad \text{ for } (i \to \infty).
\end{align*}
\end{theorem}

\begin{proof}
We have
\begin{align*}
\tfrac{\decodingradius}{n}
  &= 1-\Big[1+\underset{= \, 1-O\left(\frac{1}{i}\right)}{\underbrace{\big(1-\tfrac{1}{s_i}\big)}} \tfrac{\IntDegree}{\IntDegree+1}\Big] \underset{ = \, \gamma^\IntDegree+O\left(\frac{1}{i}\right)}{\underbrace{\tfrac{\binom{\IntDegree+\round{\gamma i}}{\IntDegree}}{\binom{\IntDegree+i}{\IntDegree}}}} - \tfrac{\IntDegree}{\IntDegree+1} \underset{\substack{= \, \gamma^{-1}\\ + O\left( \frac{1}{i} \right)}}{\underbrace{\tfrac{\ell_i}{s_i}}}\tfrac{\mH}{n} 
+ \underset{= \, O\left(\tfrac{1}{i}\right)}{\underbrace{\tfrac{1}{s_i}}} \underset{= \,-1+O\left(\frac{1}{i}\right)}{\underbrace{\left[\tfrac{1}{\CountLeq{\ell_i}}-1\right]}} \\
&= 1+\tfrac{m}{m+1} \underset{= \, 0}{\underbrace{\Big( \gamma^\IntDegree-\gamma^{-1} \tfrac{\mH}{n} \Big)}} - \gamma^\IntDegree - O\left( \tfrac{1}{i} \right) = 1 - \left(\tfrac{\mH}{n}\right)^{\frac{\IntDegree}{\IntDegree+1}} - O\left( \tfrac{1}{i} \right),
\end{align*}
which proves the claim.
\end{proof}

Note that the choice of $\ell_i$ and $s_i$ in Theorem~\ref{thm:IH_asymptotic_relative_decoding_radius} ensures that $\tau+\ell_i \mH < s_in$ for all $\tau \leq s \cdot \decodingradius(\ell_i,s_i)$.

\section{Numerical Results}
\label{sec:numerical_results}

In this section, we present simulation results.
We have conducted Monte-Carlo simulations for estimating the failure probability of the new decoding algorithm in a channel that randomly adds $t = |\Eset|$ errors, using $N \in \{10^3,10^4\}$ samples.
The decoder was implemented in SageMath~v7.5 \cite{stein_sagemath_????}, based on the power decoder implementation of \cite{nielsen2015sub}.

All simulated examples fulfill $\degH \Lambda_s \geq st \geq 2g-1$.
If this condition is not fulfilled, the simulation results might differ from the expected decoding radius, cf.~Remark~\ref{rem:weierstrass_gaps_special_case}.

\subsection{Case $\IntDegree=1$}
\label{ssec:sim_h=1}

We first compare the new improved power ($\PfailPower$) with the Guruswami--Sudan ($\PfailGS$) decoder.
The used implementation of the Guruswami--Sudan decoder is the publicly available one from \cite{nielsen2015sub}.
Table~\ref{tab:failure_rate} presents the simulation results for various code ($q,m,n,k,\ddesigned$), decoder ($\ell,s$), and channel ($t$) parameters.

\begin{table}[h!]
\caption{Observed failure rate of the improved power ($\PfailPower$) and Guruswami--Sudan ($\PfailGS$) decoder for $h=1$. Code parameters $q,m,n,k,\ddesigned$. Decoder parameters $\ell,s$. Number of errors $t$ (${^+}t=\decodingradius$ decoding radius as in \eqref{eq:decoding_radius}). Number of experiments $N$.}
\label{tab:failure_rate}
\centering
{
\renewcommand{\arraystretch}{1.2}
\setlength{\tabcolsep}{5pt}
\begin{tabular}{c|c||c|c|c||c|c||c|c|c||c}
$q$ & $m$ & $n$ & $k$ & $\ddesigned$ & $\ell$ & $s$ & $t$ & $\PfailPower$ & $\PfailGS$  & $N$ \\
\hline \hline
$4$ & $15$ &  $64$ & $10$ &  $49$ & $4$ & $2$ & $28$ & $0$ & $0$ & $10^4$ \\
    &      &       &      &       &     &     & $\phantom{{}^+}29^+$ & $0$ & $3.30 \cdot 10^{-3}$ & $10^4$ \\
    &      &       &      &       &     &     & $30$ & $9.93 \cdot 10^{-1}$ & $9.39 \cdot 10^{-1}$ & $10^4$ \\
\hline
$5$ & $55$ & $125$ & $46$ &  $70$  & $3$ & $2$ & $35$ & $0$ & $0$ & $10^4$ \\
    &      &       &      &      &     &     & $\phantom{{}^+}36^+$ & $0$ & $4.00 \cdot 10^{-4}$ & $10^4$ \\
    &      &       &      &      &     &     & $37$ & $9.57 \cdot 10^{-1}$ & $9.60 \cdot 10^{-1}$ & $10^4$ \\
\hline
$5$ & $20$ & $125$ & $11$ & $105$ & $5$ & $2$ & $67$ & $0$ & $0$ & $10^3$ \\
    &      &       &      &       &     &     & $\phantom{{}^+}68^+$ & $0$ & $7.00 \cdot 10^{-3}$ & $10^3$ \\
    &      &       &      &       &     &     & $69$ & $9.91 \cdot 10^{-1}$ & $9.60 \cdot 10^{-1}$ & $10^3$ \\
\hline
$7$ & $70$ & $343$ & $50$ & $273$ & $3$ & $2$ & $160$ & $0$ & $0$ & $10^3$ \\
    &      &       &      &       &     &     & $\phantom{{}^+}161^+$ & $0$ & $0$ & $10^3$ \\
    &      &       &      &       &     &     & $162$ & $9.78 \cdot 10^{-1}$ & $9.86 \cdot 10^{-1}$ & $10^3$ \\
\hline
$7$ & $70$ & $343$ & $50$ & $273$ & $4$ & $2$ & $168$ & $0$ & $0$ & $10^3$ \\
    &      &       &      &       &     &     & $\phantom{{}^+}169^+$ & $0$ & $0$ & $10^3$ \\
    &      &       &      &       &     &     & $170$ & $9.79 \cdot 10^{-1}$ & $2.2 \cdot 10^{-2}$ & $10^3$ \\
\hline 
$7$ & $55$ & $343$ & $35$ & $288$ & $4$ & $2$ & $183$ & $0$ & $0$ &$10^3$ \\
    &      &       &      &       &     &     & $\phantom{{}^+}184^+$ & $0$ & $0$ & $10^3$ \\
    &      &       &      &       &     &     & $185$ & $9.82 \cdot 10^{-1}$ & $1.9 \cdot 10^{-2}$ & $10^3$
\end{tabular}
}
\end{table}

It can be observed that both algorithms can almost always correct $\decodingradius$ many errors, improving upon classical power decoding.
Also, neither of the two algorithms is generally superior in terms of failure probability.

When comparing the two algorithms, one has to keep in mind that the GS algorithm is guaranteed to work only up to $n [1-\tfrac{s+1}{2(\ell+1)}] - \tfrac{\ell}{2s} \mH - \tfrac{g}{s} = \decodingradius - \tfrac{g}{s} + \tfrac{\ell}{s(\ell+1)}$ errors.

\subsection{General Case}
\label{ssec:sim_general}

We now turn to the general case of $h > 1$, where the previous best relative decoding radius is $t_{\mathrm{K}} = \tfrac{\IntDegree}{\IntDegree+1}(n-\mH)$ \cite{kampf2014bounds}.
The simulations results for various code and decoder parameters are given in Table~\ref{tab:IH_failure_rate}.

\begin{table}[h!]
\caption{Observed failure rate of Algorithm~\ref{alg:IH_improved_power} ($\PfailPower$) for $h>1$. Code parameters $q,\mH,n,k,\ddesigned,\IntDegree$. Decoder parameters $\ell,s$. Number of errors $t$ (${^+}t=\decodingradius$ decoding radius as in \eqref{eq:decoding_radius}). Number of experiments $N$. Previous best decoding radius $t_\mathrm{K}$ \cite{kampf2014bounds}.}
\label{tab:IH_failure_rate}
\centering
{
\renewcommand{\arraystretch}{1.2}
\setlength{\tabcolsep}{5pt}
\begin{tabular}{c|c||c|c|c|c||c|c||c|c||c||c}
$q$ & $\mH$ & $n$ & $k$ & $\ddesigned$ & $\IntDegree$ & $\ell$ & $s$ & $t$ & $\PfailPower$ & $N$ & $t_\mathrm{K}$ \\
\hline \hline
$4$ & $15$ &  $64$ & $10$ &  $49$ & $2$ & $3$ & $2$ & $\phantom{{}^+}35^+$ & $0$ & $10^5$ & $32$ \\
    &      &       &      &       &     &     &     & $36$ & $9.18 \cdot 10^{-1}$ & $10^3$ & $32$ \\
\hline
$4$ & $15$ &  $64$ & $10$ &  $49$ & $2$ & $5$ & $3$ & $\phantom{{}^+}36^+$ & $0$ & $10^3$ & $32$ \\
    &      &       &      &       &     &     &     & $37$ & $9.31 \cdot 10^{-1}$ & $10^3$ & $32$ \\
\hline
$4$ & $15$ &  $64$ & $10$ &  $49$ & $3$ & $3$ & $2$ & $\phantom{{}^+}38^+$ & $0$ & $10^5$ & $36$ \\
    &      &       &      &       &     &     &     & $39$ & $9.42 \cdot 10^{-1}$ & $10^3$ & $36$ \\
\hline
$4$ & $15$ &  $64$ & $10$ &  $49$ & $3$ & $4$ & $3$ & $\phantom{{}^+}39^+$ & $0$ & $10^2$ & $36$ \\
    &      &       &      &       &     &     &     & $40$ & $1$ & $10^2$ & $36$ \\
\hline
$4$ & $22$ &  $64$ & $17$ &  $42$ & $2$ & $4$ & $3$ & $\phantom{{}^+}29^+$ & $0$ & $10^3$ & $28$ \\
    &      &       &      &       &     &     &     & $30$ & $9.44 \cdot 10^{-1}$ & $10^3$ & $28$ \\
\hline \hline
$5$ & $20$ & $125$ & $11$ & $105$ & $2$ & $3$ & $2$ & $\phantom{{}^+}79^+$ & $0$ & $10^5$ & $70$ \\
    &      &       &      &       &     &     &     & $80$ & $9.37 \cdot 10^{-1}$ & $10^3$ & $70$ \\
\hline
$5$ & $20$ & $125$ & $11$ & $105$ & $2$ & $4$ & $2$ & $\phantom{{}^+}81^+$ & $0$ & $10^3$ & $70$ \\
    &      &       &      &       &     &     &     & $82$ & $9.93 \cdot 10^{-1}$ & $10^3$ & $70$ \\
\hline
$5$ & $20$ & $125$ & $11$ & $105$ & $3$ & $3$ & $2$ & $\phantom{{}^+}86^+$ & $0$ & $10^3$ & $78$ \\
    &      &       &      &       &     &     &     & $87$ & $9.94 \cdot 10^{-1}$ & $10^3$ & $78$ \\
\hline \hline
$5$ & $55$ & $125$ & $46$ & $70$  & $2$ & $4$ & $3$ & $\phantom{{}^+}48^+$ & $0$ & $10^3$ & $46$ \\
    &      &       &      &       &     &     &     & $49$ & $9.86 \cdot 10^{-1}$ & $10^3$ & $46$ \\
\hline \hline
$7$ & $90$ & $343$ & $70$ & $253$ & $2$ & $3$ & $2$ & $\phantom{{}^+}183^+$ & $0$ & $10^3$ & $168$ \\
    &      &       &      &       &     &     &     & $184$ & $9.72 \cdot 10^{-1}$ & $10^3$ & $168$ \\
\hline \hline
$8$ & $128$ & $512$ & $101$ & $384$ & $2$ & $3$ & $2$ & $\phantom{{}^+}281^+$ & $0$ & $10^2$ & $256$ \\
    &      &       &      &       &     &     &     & $282$ & $1$ & $10^2$ & $256$ \\
\end{tabular}
}
\end{table}

In all tested cases, Algorithm~\ref{alg:IH_improved_power} corrected all decoding trials up to $\decodingradius$ many errors and failed with large observed probability one error beyond this radius.

\section{Efficiently Finding a Minimal Solution of Problem~\ref{prob:IH_decoding_SRP}}\label{sec:complexity}

We use the $\FX$-vector representation of an element of $\Rback$ (cf.~\cite{nielsen2015sub}) to reformulate Problem~\ref{prob:IH_decoding_SRP} over $\FX$.
Recall that for $a \in \Rback$, we can write $a = \sum_{i=0}^{q-1} a_i Y^i \in \Rback$ with unique $a_i \in \FX$. Then, the \emph{vector representation} \cite{nielsen2015sub} of $a$ is defined by $\vr(a) = (a_0,\dots,a_{q-1}) \in \FX^q$.
Note that $q\deg(a_i)+i(q+1)\leq \degH(a)$. For $a,b \in \Rback$ it can be shown that
\begin{align*}
\vr(a+b) = \vr(a) + \vr(b), \qquad
\vr(ab) = \vr(a) \mr(b) \XiM, 
\end{align*}
where $\mr(b) \in \FX^{q \times (2q-1)}$ and $\XiM \in \FX^{(2q-1) \times q}$ are defined by
\begin{align*} \small
\mr(b) := 
\begin{bmatrix}
b_0 & b_1 & b_2 & \dots & b_{q-1} & & &\\
    & b_0 & b_1 & \dots & b_{q-2} & b_{q-1} & &\\
    &     & \ddots & \ddots & \dots & \ddots & \ddots & \\
    &     &        & b_0    & b_1 & \dots & b_{q-2} & b_{q-1} 
\end{bmatrix},
\quad
\XiM
:=
\begin{bmatrix}
1 &   &        &    \\
  & 1 &        &    \\
  &   & \ddots &    \\
  &   &        & 1  \\
X^{q+1}  & -1  &   &    \\
 & X^{q+1}  & -1  &    \\
 & \ddots  & \ddots  &    \\
 &   & X^{q+1}   &  -1  \\
\end{bmatrix}.
\end{align*}
Note further that for $c \in \Fqtwo[X]$ we have simply $\mr(ac) = \mr(a)c$.
Using the notation above, we can reformulate Problem~\ref{prob:IH_decoding_SRP} into the following problem over $\Fqtwo[X]$.
In the following, let $[q)$ denote $\{0,\dots,q-1\}$.

\begin{problem}\label{prob:IH_decoding_SRP_FX}
Given positive integers $s \leq \ell$, $\R$ and $G$ as in Section~\ref{sec:key_equations}, and
\begin{align*}
\AM{\i,\j} := \mr(A_{\i,\j}) \XiM = \mr\left(\tbinom{\j}{\i} \R^{\j-\i} G^\len{\i}\right) \XiM \in \FX^{q \times q}
\end{align*}
for all $\i \in \Iset := \{\i \in \NN_0^\IntDegree : 0 \leq \len{\i} < s\}$ and $\j \in \Jset := \{\j \in \NN_0^\IntDegree : 1 \leq \len{\j} \leq \ell\}$.
Find $\lambda_{\i,\iota}, \psi_{\j,\kappa} \in \FX$ for $\i \in \Iset$, $\j \in \Jset$, $\iota,\kappa \in [q)$, not all zero, such that
\begin{align*}
\psi_{\j,\kappa} &= \sum_{\i \in \Iset} \sum_{\iota=0}^{q-1} \lambda_{\i,\iota} \AMe{\i,\j}_{\iota,\kappa} &&1 \leq \len{\j} < s, \\
\psi_{\j,\kappa} &\equiv \sum_{\i \in \Iset} \sum_{\iota=0}^{q-1} \lambda_{\i,\iota} \AMe{\i,\j}_{\iota,\kappa} \mod G^s  &&1 \leq \len{\j} < s, \\
\max_{\iota \in [q)} \left\{q \deg \lambda_{\0,\iota} + \iota(q+1) \right\} &\leq q \deg \lambda_{\i,\iota}  + \iota(q+1) - \len{\i}(2g-1) &&0 \leq \len{\i} < s, \, \iota \in [q) \\ 
\max_{\iota \in [q)} \left\{q \deg \lambda_{\0,\iota} + \iota(q+1) \right\} &\leq q \deg \psi_{\j,\kappa}  + \kappa(q+1) - \len{\j}\mH &&1 \leq \len{\j} \leq \ell, \, \kappa \in [q) \\ 
\end{align*}
\end{problem}

Similar to its $\Rback$-equivalent, we define the degree of a solution of the above problem to be $\max_{\iota \in [q)} \left\{q \deg \lambda_{\0,\iota} + \iota(q+1) \right\}$ and call the solution monic if the leading coefficient of the $\lambda_{\0,\iota}$ that maximizes $\max_{\iota \in [q)} \left\{q \deg \lambda_{\0,\iota} + \iota(q+1) \right\}$ is $1$. The following statement establishes the connection between Problem~\ref{prob:IH_decoding_SRP} and Problem~\ref{prob:IH_decoding_SRP_FX}.

\begin{theorem}
Let $\tau \in \ZZ_{\geq 0}$.
Then, $\lambda_\i,\psi_\j \in \Rback$ for $\i \in \Iset$ and $\j \in \Jset$ is a solution of degree $\tau$ of Problem~\ref{prob:IH_decoding_SRP} if and only if
\begin{align*}
[\lambda_{\i,0}, \dots, \lambda_{\i,q-1}] &:= \vr(\lambda_\i) \\
[\psi_{\j,0}, \dots, \psi_{\j,q-1}] &:= \vr(\psi_\j)
\end{align*}
is monic solution of degree $\tau$ of Problem~\ref{prob:IH_decoding_SRP_FX}.
\end{theorem}
\begin{proof}
  If $\lambda_\i$ and $\psi_\j$ forms a solution to Problem~\ref{prob:IH_decoding_SRP} this means for $\len{\j} \geq s$ that there is some $u_\j \in \Rback$ such that:
  \begin{align*}
    \psi_\j &= \sum_{\substack{\i \smaller \j \\ \len{\i} < s}} \lambda_\i A_{\i,\j} + u_\j G^s  &\iff \\
    \vr(\psi_\j) &= \sum_{\substack{\i \smaller \j \\ \len{\i} < s}} \vr(\lambda_\i)\mr(A_{\i,\j})\XiM + \mr(u_\j) G^s \ .
  \end{align*}
  since $G^s \in \Fqtwo[X]$.
  This implies element-wise the congruence of Problem~\ref{prob:IH_decoding_SRP_FX}.
  The opposite direction is analogous, as is the case $\len{\j} < s$.
  The degree restrictions follow immediately from $\degH(a) = \max_{\iota \in [q)}\{ q \deg a_\iota + \iota(q+1) \}$ for any $a \in \Rback$ and $\vr(a) = (a_1,\ldots,a_{q-1})$.
\end{proof}

Problem~\ref{prob:IH_decoding_SRP_FX} is of the same instance as the problem discussed in \cite[Section~V.B]{nielsen2015sub},\footnote{Using this approach, it is necessary to reformulate the equations for $1 \leq \len{\j} < s$ into congruences modulo $x^\xi$, where $\xi$ is greater than the largest possible degree of the $\lambda_{\i,\iota} \AMe{\i,\j}_{\iota,\kappa}$ for the maximal number of errors that we expect to corrected.} which can be solved by transforming an $\FX$-module basis that depends on the entries of $\AM{\i,\j}$ and the relative degree bounds in Problem~\ref{prob:IH_decoding_SRP_FX}, into a reduced polynomial matrix form (weak Popov form). Using this approach, finding a minimal solution of Problem~\ref{prob:IH_decoding_SRP} can be implemented in
\begin{align*}
\Osoft\left( \tbinom{\IntDegree+\ell}{\IntDegree}^{\omega} s n^{\frac{\omega+2}{3}}  \right) \subseteq  \Osoft\left( \ell^{\IntDegree\omega} s n^{\frac{\omega+2}{3}} \right)
\end{align*}
operations over $\Fqtwo$, where $2 \leq \omega \leq 3$ is the matrix multiplication exponent.
In \cite{nielsen2016power}, a similar kind of problem was reduced to finding solution bases of so-called Pad\'e approximation problems. In this way, the complexity can be slightly reduced to
\begin{align*}
\Osoft\left( \tbinom{\IntDegree+s-1}{\IntDegree} \tbinom{\IntDegree+\ell}{\IntDegree}^{\omega-1} s n^{\frac{\omega+2}{3}}  \right) \subseteq  \Osoft\left( \ell^{\IntDegree(\omega-1)} s^{\IntDegree+1} n^{\frac{\omega+2}{3}} \right)
\end{align*}
operations over $\Fqtwo$.
In order to achieve the asymptotic decoding radius, the code parameters must be chosen as in Section~\ref{ssec:asymptotic_anaylsis}. In this case, the two asymptotic complexity statements above coincide and we get the following result.
\begin{theorem}
For a fixed code of rate $R = \tfrac{k}{n}$ and any constant $\varepsilon > 0$, we can choose $s, \ell \in O(1/\varepsilon)$ such that $\decodingradius \geq n(1 - (R+\tfrac{g-1}{n})^{\frac{\IntDegree}{\IntDegree+1}} - \varepsilon)$.
In this case, Algorithm~\ref{alg:IH_improved_power} costs $O^\sim((1/\varepsilon)^{\IntDegree \omega + 1} n^{\frac{\omega+2}{3}})$.
\end{theorem}

\begin{proof}
The first statement directly follows from Theorem~\ref{thm:IH_asymptotic_relative_decoding_radius}.
The pre- and post-computations in Algorithm~\ref{alg:IH_improved_power} are negligible compared to Line~\ref{line:minimal_solution} by similar arguments as in \cite{nielsen2015sub}.
The complexity thus follows by the arguments above.
\end{proof}

\section{Comparison to Interleaved Reed--Solomon Codes}\label{sec:comparison}

An $h$-interleaved code over some field $\mathbb{F}_Q$ over the burst error channel can equivalently be considered as a code over $\mathbb{F}_{Q^h}$ considered over the $Q^h$-ary channel.
This allows comparing the decoding capability of interleaved $1$-H codes with other constructions of short codes over large fields, most notably RS codes and interleaved RS codes, see Figure~\ref{fig:comparison_IRS_IH} for the case $Q^h = q^6$.

More precisely, for any $h \in \ZZ_{>0}$, we have several ways of obtaining $[n,k]$ codes over $\mathbb{F}_{q^{6h}}$ for $n = q^3$ and some dimension $k < n$.
We will compare the following relative decoding radii:
\begin{description}
  \item[$t_\mathrm{RS}$:] an RS code over $\mathbb{F}_{q^{6h}}$ decoding up to the Johnson radius using one of \cite{guruswami1998improved,wu_new_2008,nielsen2016power}.\footnote{%
    As pointed out in \cite{sidorenko_decoding_2008}, an RS code over $\mathbb{F}_{q^{6h}}$ whose evaluation points all lie in $\mathbb{F}_{q^3}$ are equivalent to a $2h$-interleaved RS codes over $\mathbb{F}_{q^3}$, i.e.~can be decoded up to $t_{\mathrm{IRS}}$.
    In our comparison here we therefore consider RS codes with arbitrary evaluation points for which this equivalence doesn't hold.
  }
  \item[$t_{IRS}$:] $2h$-interleaved RS code over $\mathbb{F}_{q^3}$ using  one of \cite{cohn2013approximate,puchinger2017decoding}.  \label{enum:comp_irs}
  \item[$t_{1H}$:] $3h$-interleaved $1$-H code over $\Fqtwo$ using the proposed algorithm.
\end{description}

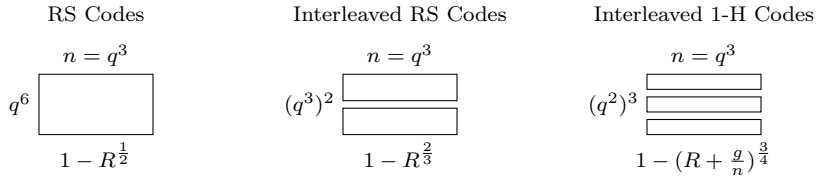
\begin{figure}[h!]
\begin{center}
\begin{tikzpicture}
\def\xdist{2.25}
\def\ydist{1.5}
\def\exampleydist{-1.5}
\def\Intydist{-5}
\def\Intxdist{\xIRS-3.7}
\def\xRS{-8}
\def\xIRS{-4}
\def\yRS{0}
\def\cwwidth{1.5}
\def\cwheight{0.2}
\def\cwdist{0.3}

\node[align=center] (H1) at (0,-0.2*\ydist) {Interleaved $1$-H Codes};
\node[align=center] (R1) at (\xIRS+0,\yRS-0.2*\ydist) {Interleaved RS Codes};
\node[align=center] (R) at (\xRS+0,\yRS-0.2*\ydist) {RS Codes};

\draw (\xRS-0.5*\cwwidth,\exampleydist-0.5*\cwheight-\cwdist) rectangle (\xRS+0.5*\cwwidth,\exampleydist+0.5*\cwheight+\cwdist);
\node[left] at (\xRS-0.5*\cwwidth,\exampleydist) {$q^6$};
\node[above] at (\xRS,\exampleydist+0.5*\cwheight+\cwdist) {$n=q^3$};
\node[below] at (\xRS,\exampleydist-0.5*\cwheight-\cwdist) {$1-R^{\frac{1}{2}}$};

\draw (\xIRS-0.5*\cwwidth,\exampleydist-1.25*\cwheight+\cwdist) rectangle (\xIRS+0.5*\cwwidth,\exampleydist+0.5*\cwheight+\cwdist);
\draw (\xIRS-0.5*\cwwidth,\exampleydist-0.5*\cwheight-\cwdist) rectangle (\xIRS+0.5*\cwwidth,\exampleydist+1.25*\cwheight-\cwdist);
\node[left] at (\xIRS-0.5*\cwwidth,\exampleydist) {$(q^3)^2$};
\node[above] at (\xIRS,\exampleydist+0.5*\cwheight+\cwdist) {$n=q^3$};
\node[below] at (\xIRS,\exampleydist-0.5*\cwheight-\cwdist) {$1-R^{\frac{2}{3}}$};

\draw (0-0.5*\cwwidth,\exampleydist-0.5*\cwheight+\cwdist) rectangle (0+0.5*\cwwidth,\exampleydist+0.5*\cwheight+\cwdist);
\draw (0-0.5*\cwwidth,\exampleydist-0.5*\cwheight)         rectangle (0+0.5*\cwwidth,\exampleydist+0.5*\cwheight);
\draw (0-0.5*\cwwidth,\exampleydist-0.5*\cwheight-\cwdist) rectangle (0+0.5*\cwwidth,\exampleydist+0.5*\cwheight-\cwdist);
\node[left] at (0-0.5*\cwwidth,\exampleydist) {$(q^2)^3$};
\node[above] at (0,\exampleydist+0.5*\cwheight+\cwdist) {$n=q^3$};
\node[below] at (0,\exampleydist-0.5*\cwheight-\cwdist) {$1-(R+\tfrac{g}{n})^{\frac{3}{4}}$};

\end{tikzpicture}
\end{center}
\caption{Example: Comparison of interleaved Reed--Solomon and one-point Hermitian codes of length $n=q^3$ over an overall field size of $q^6$.}
\label{fig:comparison_IRS_IH}
\end{figure}

These values are as follows:
\begin{align*}
  t_\mathrm{RS} &= 1-(\tfrac{k-1}{n})^{\frac 1 2}
                     &
  t_\mathrm{IRS} &= 1-(\tfrac{k-1}{n})^{\frac{2h}{2h+1}},
                      &
t_\mathrm{IH} &= 1-(\tfrac{k-1}{n}-\tfrac{g}{n})^{\frac{3h}{3h+1}},
\end{align*}
The asymptotics are already clear: since $\tfrac{g}{n} \to 0$ for $n \to \infty$, we can asymptotically achieve larger decoding radii with interleaved $1$-H codes than with interleaved RS codes, when considering comparable overall field size.
Below follows some concrete parameter examples.

$q=13$ is the smallest prime power for which $t_{\mathrm{IH}} > t_{\mathrm{IRS}}$ for rate $1/2$, i.e.~both interleaved codes can be considered as $[2197,1098]$ codes over $\mathbb{F}_{13^6}$, and the decoding radii are $t_\mathrm{RS} = 644$, $t_\mathrm{IRS} = 814$ and $t_\mathrm{IH} = 823$.
A list of decoding radii of rate $1/2$ codes with even $q$ is given in Table~\ref{tab:example_comparison_tauIH_tauIRS}.

\begin{table}[h]
\caption{Examples for $t_\mathrm{IH}$ and $t_\mathrm{IRS}$ for rate $1/2$ codes of several lengths for $\IntDegreeRS = 2$ and $\IntDegreeH = 3$.}
\label{tab:example_comparison_tauIH_tauIRS}
\centering
{
\renewcommand{\arraystretch}{1.2}
\setlength{\tabcolsep}{5pt}
\begin{tabular}{c|c|c||c|c|c||c|c}
$q$     & $n=q^3$    & $k = \tfrac{n}{2}$ & $t_\mathrm{RS}$ & $t_\mathrm{IRS}$ & $t_\mathrm{IH}$ & $t_\mathrm{IH}/t_\mathrm{RS} \approx$ & $t_\mathrm{IH}/t_\mathrm{IRS} \approx$ \\
\hline
$2^{3}$ & $     512$ & $     256$         & $     150$      & $     190$       & $     183$      & $1.22$                                & $0.96$                                 \\ 
$2^{4}$ & $    4096$ & $    2048$         & $    1200$      & $    1516$       & $    1555$      & $1.30$                                & $1.03$                                 \\ 
$2^{5}$ & $   32768$ & $   16384$         & $    9598$      & $   12126$       & $   12844$      & $1.34$                                & $1.06$                                 \\ 
$2^{6}$ & $  262144$ & $  131072$         & $   76780$      & $   97004$       & $  104478$      & $1.36$                                & $1.08$                                 \\ 
$2^{7}$ & $ 2097152$ & $ 1048576$         & $  614242$      & $  776029$       & $  842936$      & $1.37$                                & $1.09$
\end{tabular}
}
\end{table}

\section{Conclusion}

We have presented a new decoding algorithm for $\IntDegree$-interleaved one-point Hermitian codes based on the improved power decoder for Reed--Solomon codes in \cite{nielsen2016power}, its generalization to $\IntDegree$-interleaved Reed--Solomon codes in \cite{puchinger2017decoding}, and the power decoder for one-point Hermitian codes in \cite{kampf2014bounds,nielsen2015sub}.

The maximal decoding radius of the new algorithm is $n(1-(R+\tfrac{g-1}{n})^{\frac{\IntDegree}{\IntDegree+1}}-\varepsilon)$ at a cost of $O^\sim((1/\varepsilon)^{\IntDegree \omega + 1} n^{\frac{\omega+2}{3}})$ operations over $\Fqtwo$, where $2 \leq \omega \leq 3$ is the matrix multiplication exponent, and improves upon previous best decoding radii at all rates.
Experimental results indicate that the algorithm achieves this maximal decoding radius with large probability.

For large $n$, interleaved one-point Hermitian codes achieve larger maximal decoding radii than interleaved Reed--Solomon codes when compared for the same length and overall field size.

In the case $\IntDegree=1$, we obtain a one-point Hermitian codes equivalent of the improved power decoder for Reed--Solomon codes in \cite{nielsen2016power}, which achieves a similar decoding radius as the Guruswami--Sudan list decoder. Simulation results indicate that the new decoder has a similar failure probability for numbers of errors beyond the latter's guaranteed decoding radius.

As for any other power decoding algorithm, both for Reed--Solomon and one-point Hermitian codes, deriving analytic bounds on the failure probability remains an open problem. So far, the only parameters for which such an expression is known are $\IntDegree=1$, $\ell \leq 3$, and $s \leq 2$, cf.~\cite{schmidt2010syndrome,nielsen2016power}.

\bibliographystyle{IEEEtran}
\bibliography{main}

\begin{thebibliography}{10}
\providecommand{\url}[1]{#1}
\csname url@samestyle\endcsname
\providecommand{\newblock}{\relax}
\providecommand{\bibinfo}[2]{#2}
\providecommand{\BIBentrySTDinterwordspacing}{\spaceskip=0pt\relax}
\providecommand{\BIBentryALTinterwordstretchfactor}{4}
\providecommand{\BIBentryALTinterwordspacing}{\spaceskip=\fontdimen2\font plus
\BIBentryALTinterwordstretchfactor\fontdimen3\font minus
  \fontdimen4\font\relax}
\providecommand{\BIBforeignlanguage}[2]{{%
\expandafter\ifx\csname l@#1\endcsname\relax
\typeout{** WARNING: IEEEtran.bst: No hyphenation pattern has been}%
\typeout{** loaded for the language `#1'. Using the pattern for}%
\typeout{** the default language instead.}%
\else
\language=\csname l@#1\endcsname
\fi
#2}}
\providecommand{\BIBdecl}{\relax}
\BIBdecl

\bibitem{guruswami1998improved}
V.~Guruswami and M.~Sudan, ``{Improved Decoding of Reed--Solomon and
  Algebraic-Geometric Codes},'' in \emph{IEEE Annual Symposium on Foundations
  of Computer Science}, 1998, pp. 28--37.

\bibitem{schmidt2010syndrome}
G.~Schmidt, V.~R. Sidorenko, and M.~Bossert, ``{Syndrome Decoding of
  Reed--Solomon Codes Beyond Half the Minimum Distance Based on Shift-Register
  Synthesis},'' \emph{IEEE Transactions on Information Theory}, vol.~56,
  no.~10, pp. 5245--5252, 2010.

\bibitem{kampf2012decoding}
S.~Kampf, ``{Decoding Hermitian Codes - An Engineering Approach},'' Ph.D.
  dissertation, Universit{\"a}t Ulm, 2012.

\bibitem{nielsen2015sub}
J.~S.~R. Nielsen and P.~Beelen, ``{Sub-Quadratic Decoding of One-Point
  Hermitian Codes},'' \emph{IEEE Transactions on Information Theory}, vol.~61,
  no.~6, pp. 3225--3240, 2015.

\bibitem{nielsen2016power}
J.~Rosenkilde, ``{Power Decoding Reed--Solomon Codes up to the Johnson
  Radius},'' \emph{Accepted for: Advances in Mathematics of Communications},
  (2018), arXiv preprint arXiv:1505.02111.

\bibitem{kampf2014bounds}
S.~Kampf, ``{Bounds on Collaborative Decoding of Interleaved Hermitian Codes
  and Virtual Extension},'' \emph{Designs, codes and cryptography}, vol.~70,
  no. 1-2, pp. 9--25, 2014.

\bibitem{krachkovsky1997decoding}
V.~Y. Krachkovsky and Y.~X. Lee, ``{Decoding for Iterative Reed--Solomon Coding
  Schemes},'' \emph{IEEE Trans.\ Magn.}, vol.~33, no.~5, pp. 2740--2742, 1997.

\bibitem{parvaresh2004multivariate}
F.~Parvaresh and A.~Vardy, ``{Multivariate Interpolation Decoding Beyond the
  Guruswami--Sudan Radius},'' in \emph{Proceedings of the 42nd Allerton
  Conference on Communication, Control and Computing}, 2004.

\bibitem{schmidt2007enhancing}
G.~Schmidt, V.~Sidorenko, and M.~Bossert, ``{Enhancing the Correcting Radius of
  Interleaved Reed-Solomon Decoding using Syndrome Extension Techniques},'' in
  \emph{IEEE ISIT}, 2007, pp. 1341--1345.

\bibitem{cohn2013approximate}
H.~Cohn and N.~Heninger, ``{Approximate Common Divisors via Lattices},''
  \emph{The Open Book Series}, vol.~1, no.~1, pp. 271--293, 2013.

\bibitem{wachterzeh2014decoding}
A.~{Wachter-Zeh}, A.~Zeh, and M.~Bossert, ``{Decoding Interleaved Reed--Solomon
  Codes Beyond Their Joint Error-Correcting Capability},'' \emph{Designs, Codes
  and Cryptography}, vol.~71, no.~2, pp. 261--281, 2014.

\bibitem{puchinger2017decoding}
S.~Puchinger and J.~{Rosenkilde n\'e Nielsen}, ``{Decoding of Interleaved
  Reed-Solomon Codes Using Improved Power Decoding},'' \emph{IEEE International
  Symposium on Information Theory}, 2017.

\bibitem{puchinger2017improved}
S.~Puchinger, I.~Bouw, and J.~{Rosenkilde n\'e Nielsen}, ``{Improved Power
  Decoding of One-Point Hermitian Codes},'' in \emph{International Workshop on
  Coding and Cryptography}, 2017, arXiv:1703.07982.

\bibitem{feng_generalized_1989}
G.-L. Feng and K.~K. Tzeng, ``A {Generalized} {Euclidean} {Algorithm} for
  {Multisequence} {Shift}-{Register} {Synthesis},'' \emph{IEEE Transactions on
  Information Theory}, vol.~35, no.~3, pp. 584--594, 1989.

\bibitem{nielsen_generalised_2013}
\BIBentryALTinterwordspacing
J.~S.~R. Nielsen, ``Generalised {Multi}-sequence {Shift}-{Register} {Synthesis}
  using {Module} {Minimisation},'' in \emph{{IEEE} {International} {Symposium}
  on {Information} {Theory}}, 2013, pp. 882--886. [Online]. Available:
  \url{http://ieeexplore.ieee.org/xpls/abs_all.jsp?arnumber=6620353}
\BIBentrySTDinterwordspacing

\bibitem{beckermann_uniform_1994}
\BIBentryALTinterwordspacing
B.~Beckermann and G.~Labahn, ``A {Uniform} {Approach} for the {Fast}
  {Computation} of {Matrix}-{Type} {Padé} {Approximants},'' \emph{SIAM Journal
  on Matrix Analysis and Applications}, vol.~15, no.~3, pp. 804--823, Jul.
  1994. [Online]. Available:
  \url{http://epubs.siam.org/doi/abs/10.1137/S0895479892230031}
\BIBentrySTDinterwordspacing

\bibitem{stein_sagemath_????}
W.~A. Stein \emph{et~al.}, ``{SageMath} {Software},'' http://www.sagemath.org.

\bibitem{wu_new_2008}
Y.~Wu, ``New {List} {Decoding} {Algorithms} for {Reed}-{Solomon} and {BCH}
  {Codes},'' \emph{IEEE Transactions on Information Theory}, vol.~54, no.~8,
  pp. 3611--3630, 2008.

\bibitem{sidorenko_decoding_2008}
V.~Sidorenko, G.~Schmidt, and M.~Bossert, ``{Decoding Punctured
  {Reed}--{Solomon} Codes up to the {Singleton} Bound},'' in
  \emph{International {ITG} {Conference} on {Source} and {Channel} {Coding}},
  2008, pp. 1--6.

\end{thebibliography}

\end{document}